\documentclass[12pt]{article}
\usepackage[left=3cm, right=3cm, top=3cm, bottom=3cm]{geometry}
\usepackage{amsmath,amsthm,amssymb,natbib,enumerate,xcolor,graphicx}
\usepackage[hyphens]{url}
\usepackage{hyperref}
\hypersetup{
    colorlinks,
    linkcolor={red!50!black},
    citecolor={blue!50!black},
    urlcolor={blue!50!black}
}
\newtheorem{thm}{Theorem}
 
\newtheorem{lemma}[thm]{Lemma}

\theoremstyle{definition}

\theoremstyle{remark}

\newcommand*{\set}[1]{\left\{#1\right\}}

\newcommand*{\x}{X}
\newcommand*{\y}{Y}

\title{Greater search cost reduces prices} 
\author{Sander Heinsalu\thanks{Research School of Economics, Australian National University, 
25a Kingsley St, Acton ACT 2601, Australia.
Email: sander.heinsalu@anu.edu.au, 
website: \url{https://sanderheinsalu.com/}
The author thanks the audience at the MIT theory lunch for comments and suggestions. The hospitality of MIT during part of this work is greatly appreciated. 
}}
\date{\today}

\begin{document}
\maketitle

\begin{abstract}
The optimal price of each firm falls in the search cost of consumers, in the limit to the monopoly price, despite the exit of lower-value consumers in response to costlier search. Exit means that fewer inframarginal consumers remain. The decrease in marginal buyers is smaller, because part of demand is composed of customers coming from rival firms. These buyers can be held up and are not marginal. Higher search cost reduces the fraction of incoming switchers among buyers, which decreases the hold-up motive, thus the price. 




%

	
\end{abstract}


Industry associations often provide a public member directory\footnote{
Grocers: \url{http://www.agbr.com/store-locator/}, notaries: \url{https://www.thenotariessociety.org.uk/notary-search}, restoration contractors: \url{https://www.iicrc.org/page/IICRCGlobalLocator}.
} that reduces customers' search costs and informs them about the value provided by each firm, but does not reveal prices. This suggests that firms in the industry benefit from cheaper search by enough to justify maintaining and updating the searchable directory on the web. 
Easier comparison of firms seemingly increases competition and reduces profits---an intuition confirmed by most industrial organisation models. 
By contrast, this paper shows how differentiated firms can in fact increase profits by reducing search costs and making the surplus offered more comparable across firms. 

The economic forces can be elucidated in a simple model of a duopoly in price competition. Firms simultaneously set prices. Consumers privately know their valuations for both firms, which may be determined by geographic location or the preferred times for a service. Initially, each consumer is familiar with one of the firms, in the sense of observing its price. Consumers may buy immediately from their `home' firm, exit, or learn the price of the competitor. After learning, consumers may buy from either firm or exit. Firms cannot distinguish customers who buy immediately from those who first learn and then decide.\footnote{
Websites try to track buyers' browsing history to segment them into switchers and captive customers, but buyers may take countermeasures (using a VPN, the Tor browser, or searching on different devices). The segmentation may be illegal or create negative publicity, making it not worthwhile. 
} 

The consumers who learn can be held up, because their willingness to pay the search cost implies that their valuation for the firm they arrive at is above its equilibrium price. The hold-up motive increases prices. Greater search cost weakens the hold-up motive, because fewer consumers search. When a smaller fraction of a firm's demand is composed of customers switching from the competitor, the hold-up motive and price are lower. 

Another intuition is that higher search costs cause some consumers to exit who previously would have switched firms. More exit leaves fewer inframarginal consumers to firms on average, so some firm's inframarginal demand falls. If the firms are symmetric enough, then each of them receives fewer inframarginal buyers. The number of marginal consumers may rise or fall in the search cost, but this change is smaller. As the ratio of inframarginal to marginal consumers falls, so does the optimal price. The prices of the firms are strategic complements, so one firm's price cut motivates others to follow suit. 

The literature on search costs and pricing is large and mostly finds that prices increase in the search cost, as in the seminal work of \cite{diamond1971}. Exceptions assume either multiproduct or multiperiod markets or add a countervailing force (higher switching cost or firm obfuscation) to a lower search cost. The current work presents a simpler one-shot, one-product framework, with a different driving force (reduced hold-up) and a stronger result---prices strictly decrease in the search cost for any positive search cost at which some consumers still switch firms. 

In \cite{zhou2014}, multiproduct search makes products complements: a price cut on one increases demand for both, more so at a greater search cost. Thus prices may fall in the search cost. 
\cite{rhodes+zhou2019} extend this result to four firms who supply two products and may merge pairwise into two-product firms. 
Higher search cost may cause mergers, which may reduce prices due to the complementarity mechanism of \cite{zhou2014}. 

\cite{klemperer1987,klemperer1995} points out that if consumers have a switching cost after their first purchase, but not initially, then higher switching costs may reduce prices even below cost in the first period. The reason is that firms compete to lock in customers to later charge the monopoly price. The second-period prices weakly increase in the switching cost. 
 
\cite{dube+2009} show numerically and \cite{cabral2009} analytically that intermediate switching cost leads to lower prices than zero switching cost in an infinite horizon model. The incentive to cut price to `invest' in customer acquisition outweighs the incentive to `harvest' with a high price. However, for large enough switching costs, prices rise. 
\cite{cabral2016} extends these results to show that if trades have high frequency or the market structure is close to symmetric duopoly, then switching costs increase competition, but with infrequent trade or sufficiently asymmetric competitors, switching costs decrease competition. 

\cite{lal+sarvary1999} model adding a web shop to a physical store. They assume that the web shop reduces search costs but increases switching costs, because it is easy to re-order a familiar brand. This may raise prices and reduce search, driven by the higher switching cost, which outweighs the lower search cost. 

In \cite{ellison+wolitzky2012}, firms obfuscate to increase consumers' search cost. With costless obfuscation, firms exactly offset a fall in the exogenous search cost, so it does not affect prices. Thus it may be said that prices weakly increase in the search cost. 

The next section introduces the framework and derives the main result. Extensions and generalisations are discussed in Section~\ref{sec:extensions}, followed by the conclusion in Section~\ref{sec:discussion}.

\section{Horizontally differentiated duopoly}
\label{sec:model}

Two firms $i\in\set{\x,\y}$ simultaneously set prices $P_i$. 
There is a mass $1$ of consumers indexed by $v=(v_{\x},v_{\y})$, where $v_{i}\in[0,1]$ is the consumer's valuation for firm $i$'s product. Consumers privately know their valuations. Firms only have the common prior belief that $v_i$ is distributed according to the pdf $f_{i}$, which is positive with interval support. 
The corresponding cdf is denoted $F_i$. 

Independently of $v$, fraction $\mu_{\x}$ of consumers initially observe $P_{\x}$, and fraction $\mu_{\y} =1-\mu_{\x}$ observe $P_{\y}$. Call the firm whose price a consumer initially observes the \emph{initial firm} of the consumer. 

Each consumer decides whether to buy from her initial firm, learn the price of the other firm at cost $s>0$\footnote{
Zero search cost is qualitatively different (Bertrand competition). Section~\ref{sec:extensions} discusses $s=0$.  
} 
or exit. After learning, the consumer decides whether to buy from firm $\x$, firm $\y$ or exit. 

The payoff from not buying is normalised to zero. 
A consumer with valuation $v$ who buys from firm $i$ at price $P_i$ without searching obtains payoff $v_i-P_i$, but after search, obtains $v_i-P_i-s$ if buys and $-s$ if exits. 
Firm $i$ that sets price $P_i$ resulting in \emph{ex post} demand $D_i$ gets \emph{ex post} profit $\pi_i:=(P_i-c_i)D_i$, with $c_i<1$. 
W.l.o.g.\ restrict $P_i\in\left[c_i,1\right]$, because pricing below cost or above the maximal valuation of consumers is never a unique best response. 
A mixed strategy of firm $i$ is the cdf $\sigma_i$ on $[c_i,1]$. 

Equilibrium consists of firms' pricing strategies and consumer decisions such that (i) each firm maximises profits given the decisions it expects of the consumers and the rival firm, (ii) consumers choose to buy, learn or exit based on the prices they see and expect, and if they learn, then choose based on the prices they see which firm, if any, to buy from to maximise their expected payoff and (iii) the expectations of the firms and consumers are correct. 

The next section first finds the optimal decisions of consumers, which determine the demands for the firms. Then the profit-maximising prices are calculated, followed the main comparative static of prices decreasing in the search cost.

\section{Demand, profit and comparative statics}
\label{sec:results}

To solve the pricing game, start with the decisions of the consumers. These determine the demands for the firms, which are then used to find the optimal prices. 

Consumer $v$ who observes firm $j$'s price $P_j$ and expects firm $i$ to choose the pricing strategy $\sigma_i^{\mathbb{E}}$ learns $P_i$ if 
\begin{align}
\label{learn}
\int_{c_i}^{1}\max\set{0,v_{i}-P_{i}^{\mathbb{E}},v_{j}-P_{j}}d\sigma_i^{\mathbb{E}}(P_i^{\mathbb{E}})-s \geq \max\set{0,v_{j}-P_{j}}. 
\end{align} 
The right-hand side (RHS) of~(\ref{learn}) is the value of not learning---either choosing to exit (payoff zero) or to buy immediately at price $P_j$. The left-hand side (LHS) is the benefit of learning minus the search cost $s$. The benefit (the integral) reflects the options of being able to exit, buy from firm $i$ or buy from firm $j$ after learning. The consumer chooses the best of these options, thus the $\max$, and before learning, forms an expectation of the best of these options based on the pricing strategy of firm $i$ (integrates over $P_i^{\mathbb{E}}$ with respect to $\sigma_i^{\mathbb{E}}$). 

The \emph{certainty equivalent price} 
\begin{align}
\label{pice}
P_{iCE} =\int_{c_i}^{v_{i}-\max\set{0,v_{j}-P_j}} P_{i}^{\mathbb{E}} d\sigma_i^{\mathbb{E}}(P_i^{\mathbb{E}})
\end{align} 
for the learning decision of a consumer who faces $P_j$ and expects $\sigma_i^{\mathbb{E}}$ is the pure price of firm $i$ that creates the same benefit of learning as $\sigma_i^{\mathbb{E}}$. Formally, $P_{iCE}$ solves 
$\max\set{v_{i}-P_{iCE},0,v_{j}-P_j} =\int_{c_i}^1\max\set{v_{i}-P_{i}^{\mathbb{E}},0,v_{j}-P_j}d\sigma_i^{\mathbb{E}}(P_i^{\mathbb{E}})$. If $\sigma_i^{\mathbb{E}}$ is the pure $P_i^*$, then of course $P_{iCE}=P^*$. 
The interpretation of~(\ref{pice}) is that a consumer cares about firm $i$'s expected price conditional on accepting it. The consumer accepts it if it is below the consumer's value $v_i$ for firm $i$ by at least the net benefit (value minus price) of buying from $j$. 


The demand for a firm consists of consumers initially at that firm who either buy immediately or learn and then buy from that firm, and consumers initially at the rival firm who learn and switch. 
Figure~\ref{fig:D} depicts demands for each firm from customers initially at each firm (left panel: buyers initially at firm $\x$, right panel: $\y$). The marginal customers for firm $\y$ are marked as the thick blue line and the marginal buyers for $\x$ as the thick orange line. 
\begin{figure}
\caption{Demands at the pure prices $P_{\x}=P_{\x}^*=P_{\x CE}=0.6$ and $P_{\y}=P_{\y}^*=P_{\y CE}=0.45$ and search cost $s=0.1$. Left panel: consumers initially at firm $\x$, right panel: $\y$.}
\label{fig:D}
\includegraphics[width=\linewidth]{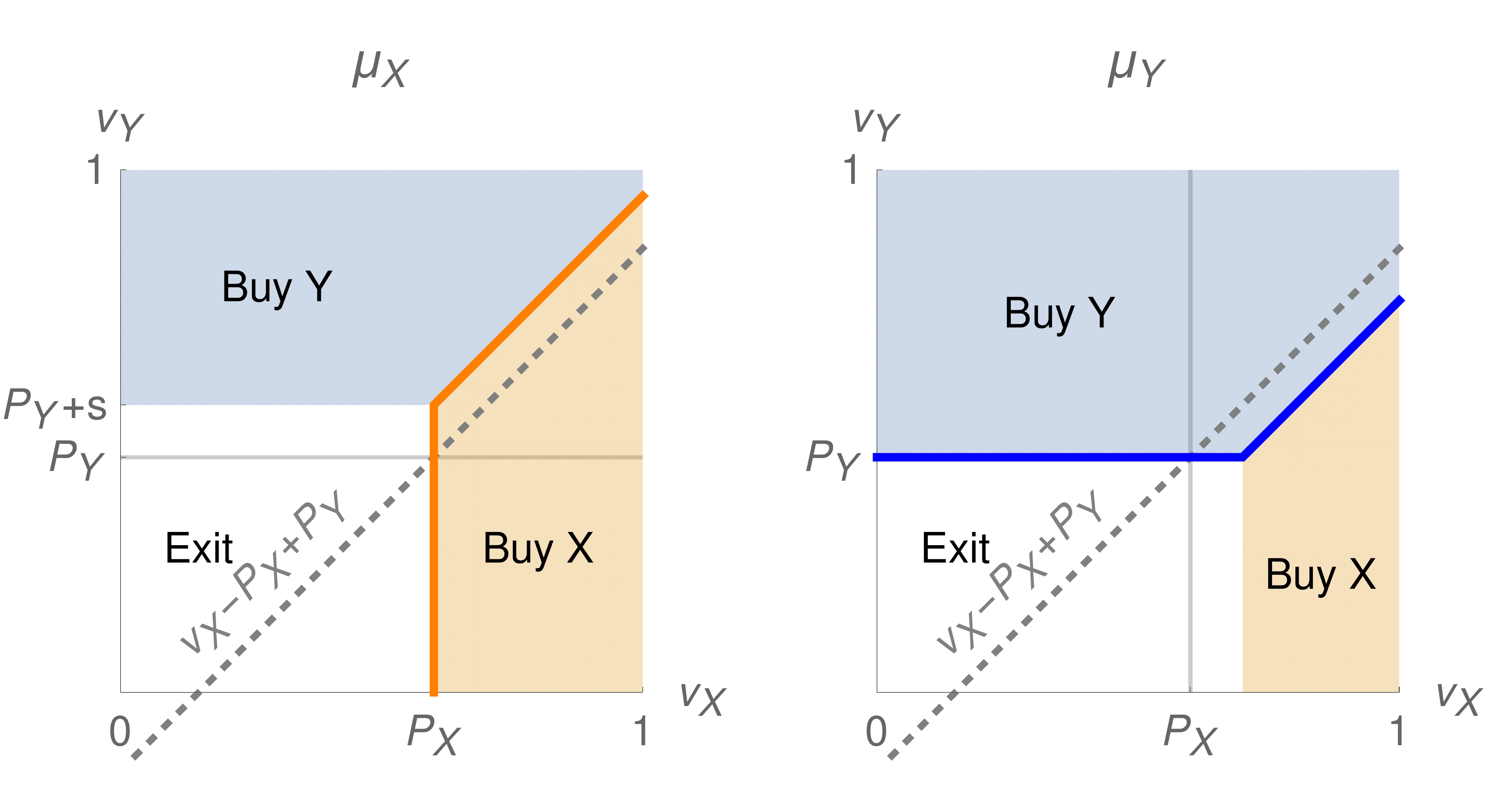} 
\end{figure}

The valuations of customers initially at $i$ who eventually buy from $i$ either motivate them to buy immediately (when $v_{i}-P_{i} \geq \max\set{0,v_{j}-P_{jCE}-s}$) or make it optimal to learn and then still buy from $i$ (when $v_{j}-P_{jCE}-s \geq \max\set{0,v_{i}-P_{i}}$ and $v_i-P_i\geq \max\set{0,v_j-P_j}$). 
These conditions combine to 
$v_{i} \geq P_i +\max\set{0,v_{j}-\max\set{P_{jCE}+s,P_j}}$. 

Consumers starting at $j$ buy from $i$ if their expected benefit from learning and the observed benefit from $i$ after learning are large enough. Formally, $v_{i}-P_{iCE}-s \geq \max\set{0,v_{j}-P_{j}}$ and $v_i-P_i\geq \max\set{0,v_j-P_j}$, which are equivalent to $v_{i} \geq \max\set{P_i,P_{iCE}+s}+\max\set{0,v_{j}-P_{j}}$.

The demand that firm $i$ expects from price $P_i$ when it expects firm $j$ to choose pricing strategy $\sigma_j^*$ and consumers initially at $j$ ($i$) to expect $\sigma_i^{\mathbb{E}}$ ($\sigma_j^{\mathbb{E}}$) 
is 
\begin{align}
\label{D}
&\notag D_{i}(P_i,\sigma_j^*,\sigma_i^{\mathbb{E}},\sigma_j^{\mathbb{E}}) 
=\mu_{i}\int_{c_j}^{1}\int_0^1\int_{P_i +\max\set{0,v_{j}-\max\set{P_{jCE}+s,P_j^*}}}^1f_i(v_{i})f_j(v_{j})dv_idv_jd\sigma_j^*(P_j^*) 
\\& +\mu_{j}\int_{c_j}^{1}\int_0^1\int_{\max\set{P_i,P_{iCE}+s}+\max\set{0,v_{j}-P_{j}^*}}^1f_i(v_{i})f_j(v_{j})dv_idv_jd\sigma_j^*(P_j^*).
\end{align} 
The inner two integrals in the demand aggregate the consumers initially at each firm over the region of valuations that result in these consumers eventually buying from $i$, given the prices. The outer integral in~(\ref{D}) reflects firm $i$'s expectation over the prices of firm $j$. 

Having derived the demand for each firm, the next preliminary lemma establishes pure best responses of the firms and the strategic complementarity of prices. 
A sufficient condition is that the densities of consumer valuations do not decrease too fast, which is satisfied by the uniform distribution and many others. The formal condition bounds below the elasticity of the pdf of valuations. 
\begin{lemma}
\label{lem:pure}
If $(P_i-c_i)\frac{\partial f_i(P_i+w)}{\partial P_i}\geq - f_i(P_i+w)$ for all $w\in[0,1]$ for each firm $i$, then each has a pure best response to any $\sigma_j^*,\sigma_i^{\mathbb{E}},\sigma_j^{\mathbb{E}}$, and prices are strategic complements. 
\end{lemma}
\begin{proof}[Proof of Lemma~\ref{lem:pure}]
Using~(\ref{D}), firm $i$'s profit has the derivative 
\begin{align}
\label{foc}
 \frac{\partial \pi_i(P_i,\sigma_j^*,\sigma_i^{\mathbb{E}},\sigma_j^{\mathbb{E}})}{\partial P_i} 
&=\int_{c_j}^{1}\int_0^1\left[1-\mu_iF_i\left(P_i +\max\set{0,v_{j}-\max\set{P_{jCE}+s,P_j^*}}\right)\right.
\\&\notag -\mu_jF_i\left(\max\set{P_i,P_{iCE}+s}+\max\set{0,v_{j}-P_{j}^*}\right)
\\&\notag -(P_i-c_i)\mu_if_i\left(P_i +\max\set{0,v_{j}-\max\set{P_{jCE}+s,P_j^*}}\right) 
\\&\notag  \left.-(P_i-c_i)\mu_j\text{\textbf{1}}_{\set{P_i>P_{iCE}+s}}f_i\left(P_i+\max\set{0, v_{j}-P_{j}^*}\right)\right]dF_j(v_j)d\sigma_j^*(P_j^*). 
\end{align}
and the second derivative 
\begin{align}
\label{soc}
\frac{\partial^2 \pi_i}{\partial P_i^2} 
&=-\left.\int_{c_j}^{1}\int_0^1 \right[ 2\mu_i f_i\left(P_i +\max\set{0,v_{j}-\max\set{P_{jCE}+s,P_j^*}}\right) 
\\&\notag +2\mu_j\text{\textbf{1}}_{\set{P_i>P_{iCE}+s}} f_i\left(P_i +\max\set{0,v_j-P_j^*} \right)  
\\&\notag  +(P_i-c_i)\mu_i \frac{\partial f_i\left(P_i +\max\set{0,v_{j}-\max\set{P_{jCE}+s,P_j^*}}\right)}{\partial P_i}
\\&\notag \left. +(P_i-c_i)\mu_j\text{\textbf{1}}_{\set{P_i>P_{iCE}+s}} \frac{\partial f_i\left(P_i +\max\set{0,v_{j}-P_j^*}\right)}{\partial P_i}\right]dF_j(v_j)d\sigma_j^*(P_j^*). 
\end{align} 
Sufficient for $\frac{\partial^2 \pi_i}{\partial P_i^2} <0\;\forall P_i\in(c_i,1)$ 
is $\int_0^1[(P_i-c_i)\frac{\partial f_i(P_i+w)}{\partial P_i} +2 f_i(P_i+w)]dF_j(v_j)\geq 0$ for all $w\in[0,1]$, which is ensured if $(P_i-c_i)\frac{\partial f_i(P_i+w)}{\partial P_i}\geq -2 f_i(P_i+w)$ for all $w\in[0,1-c_j]$. 
%
Therefore the best response (BR) of firm $i$ to any $\sigma_j^*,\sigma_i^{\mathbb{E}},\sigma_j^{\mathbb{E}}$ is pure and unique. 

Focus on pure strategies from now on, so $P_{jCE}=P_j$. By \cite{milgrom+roberts1990} Theorem 4, the game is supermodular if $\frac{\partial^2 \pi_i}{\partial P_i\partial P_j}\geq0$, in which case prices are strategic complements.  
The cross-partial derivative is  
\begin{align}
\label{crosspartial}
\frac{\partial^2 \pi_i}{\partial P_i\partial P_j}
&=\int_{P_j^*+s}^1\mu_if_i\left(P_i +v_{j}-P_j^*-s\right)dF_j(v_j)
\\&\notag +\int_{P_j^*}^1\mu_jf_i\left(\max\set{P_i,P_{iCE}+s}+v_{j}-P_{j}^*\right)dF_j(v_j)
\\&\notag +\int_{P_j^*+s}^1(P_i-c_i)\mu_i\frac{\partial f_i\left(P_i +v_{j}-P_j^*-s\right)}{\partial P_i} dF_j(v_j)
\\&\notag  +\int_{P_j^*}^1(P_i-c_i)\mu_j\text{\textbf{1}}_{\set{P_i>P_{iCE}+s}}\frac{\partial f_i\left(P_i+v_{j}-P_{j}^*\right)}{\partial P_i}dF_j(v_j)
\end{align} 
because $\frac{\partial f_i}{\partial P_j} =-\frac{\partial f_i}{\partial P_i}$. 
Sufficient for $\frac{\partial^2 \pi_i}{\partial P_i^2} <0 <\frac{\partial^2 \pi_i}{\partial P_i\partial P_j}$ is $\int_{0}^1[(P_i-c_i)\frac{\partial f_i(P_i+w)}{\partial P_i} + f_i(P_i+w)]dF_j(v_j)\geq 0$ for all $w\in[0,1]$, which is ensured if $(P_i-c_i)\frac{\partial f_i(P_i+w)}{\partial P_i}\geq - f_i(P_i+w)$ for all $w\in[0,1-c_j]$. 
%
\end{proof}
Given Lemma~\ref{lem:pure}, pure strategies are assumed from now on. The conditions in the lemma are far from necessary for pure equilibria---a unique pure best response to any strategy of the competitor is clearly much stronger than needed. Similarly, prices are strategic complements under weaker conditions, but these are more complicated. 

With strategic complementarities, the prices of the firms move together, the equilibria with the lowest and highest prices are stable, and all stable equilibria have the same comparative statics. 
A firm's profit increases in a rival's price, so firms impose positive externalities on each other by raising price. This implies that equilibria are Pareto ordered by price. The highest-price equilibrium is the natural focus of coordination if multiple equilibria exist. 

The next lemma shows that equilibrium is unique if, in addition to a weaker condition than in Lemma~\ref{lem:pure}, the consumer valuation pdf is weakly decreasing (e.g., uniform) and consumers are initially evenly distributed among firms. Firms may be asymmetric in other respects, e.g., in costs and how much consumers value their product. 
\begin{lemma}
\label{lem:unique} 
If $\mu_j=\mu_i$ and for each firm, $\frac{\partial f_i(P_i)}{\partial P_i}\leq 0$ and $(P_i-c_i)\frac{\partial f_i(P_i)}{\partial P_i}\geq -2f_i(P_i)$, then the equilibrium is unique. 
\end{lemma}
\begin{proof}[Proof of Lemma~\ref{lem:unique}]
Profit after imposing the equilibrium condition $P_i=P_{iCE}=P_i^*$ on both firms is denoted $\pi_i^*$. 
Sufficient for a unique equilibrium is 
that the slopes of best responses are below $1$. Formally, 
$\left|-\frac{\partial^2 \pi_i^*}{\partial P_i\partial P_j}\left/\frac{\partial^2 \pi_i^*}{\partial P_i^2}\right.\right| <1$ for each firm $i$. Equivalently, 
$\frac{\partial^2 \pi_i^*}{\partial P_i\partial P_j} +\frac{\partial^2 \pi_i^*}{\partial P_i^2}<0$. 

The derivatives of $\pi_i^*$ are obtained from~(\ref{soc}) and~(\ref{crosspartial}) by substituting $\text{\textbf{1}}_{\set{P_i>P_{iCE}+s}}=0$ and $\max\set{P_{jCE}+s,P_j^*} =P_j+s$. 
Then 
\begin{align*}
&\frac{\partial^2 \pi_i^*}{\partial P_i\partial P_j} +\frac{\partial^2 \pi_i^*}{\partial P_i^2} 
=\mu_i\int_{P_j}^{1} f_i(P_i+v_j-P_j-s)dF_j(v_j) +\mu_j\int_{P_j}^{1} f_i(P_i+v_j-P_j+s)dF_j(v_j) 
\\&\notag +\mu_i(P_i-c_i)\int_{P_j}^{1} \frac{\partial f_i(P_i+v_j-P_j-s)}{\partial P_i}dF_j(v_j) -2\mu_i\int_0^1 f_i\left(P_i +\max\set{0, v_{j}-P_{j}-s}\right)dF_j(v_j) 
\\&\notag -\mu_i(P_i-c_i)\int_0^1 \frac{\partial f_i\left(P_i +\max\set{0, v_{j}-P_{j}-s}\right)}{\partial P_i}dF_j(v_j) 
\\&\notag =-\mu_i\left[2f_i(P_i-s) +(P_i-c_i)\frac{\partial f_i(P_i-s)}{\partial P_i}\right]F_j(P_j) 
\\&\notag +\int_{P_j}^{1} [\mu_jf_i(P_i+v_j-P_j+s) -\mu_i f_i\left(P_i +v_{j}-P_{j}-s\right)]dF_j(v_j)
\end{align*} 
Sufficient for uniqueness is $(P_i-c_i)\frac{\partial f_i(P_i)}{\partial P_i}\geq -2f_i(P_i)$ and 
$\mu_j=\mu_i$ and $f_i' \leq 0$. 
\end{proof}



The main result establishes that if the firms have similar initial demands and the consumer valuation distribution does not vary too fast, then each firm's price decreases in the search cost of the consumers.  
Uniform valuations satisfy the condition, as does a truncated exponential distribution if search is not too costly. 
\begin{thm}
\label{thm:main}
If $\mu_if_i\left(P_i-s +w\right) +\mu_i(P_i-c_i)\frac{\partial f_i\left(P_i-s +w\right)}{\partial P_i} \leq \mu_jf_i\left(P_i+s +w\right)$ for all $w\in[0,1-s]$ and firms set pure prices which are strategic complements, 
then $\frac{dP_i^*}{ds} \leq 0$ for both firms in any stable equilibrium. 
\end{thm}
\begin{proof}[Proof of Theorem~\ref{thm:main}]
With pure prices, the FOC of firm $i$ in~(\ref{foc}) after imposing the equilibrium condition $P_i=P_{iCE}=P_i^*$ for each firm is 
\begin{align}
\label{focstar}
&FOC_i^* 
=\int_0^1\left[1-\mu_iF_i\left(P_i +\max\set{0, v_{j}-P_{j}-s}\right)\right.
\\&\notag -\mu_jF_i\left(P_i+s+\max\set{0,v_{j}-P_{j}}\right) \left.-(P_i-c_i)\mu_if_i\left(P_i +\max\set{0, v_{j}-P_{j}-s}\right)\right]dF_j(v_j). 
\end{align}
Its derivative w.r.t.\ $s$ is 
\begin{align}
\label{dfocds}
&\notag \frac{\partial FOC_i^*}{\partial s} =\mu_i\int_{P_j+s}^1\left[f_i\left(P_i +v_{j}-P_{j}-s\right) +(P_i-c_i)\frac{\partial f_i\left(P_i +v_{j}-P_{j}-s\right)}{\partial P_i} \right]dF_j(v_j) 
\\& -\mu_j\int_0^1f_i\left(P_i+s+\max\set{0,v_{j}-P_{j}}\right)dF_j(v_j),
\end{align} 
negative if 
$\mu_if_i\left(P_i-s +w\right) +\mu_i(P_i-c_i)\frac{\partial f_i\left(P_i-s +w\right)}{\partial P_i} \leq \mu_jf_i\left(P_i+s +w\right)$ for all $w\in[0,1-s]$. 
For specific distributions,~(\ref{dfocds}) can be calculated explicitly. Sufficient for $\frac{\partial FOC_i^*}{\partial s}<0$ is that
$f_i$ is uniform and $\mu_i(1-P_j-s)\leq \mu_j(1-P_j)$ (with uniform distributions, $P_j\geq\frac{1}{2\sqrt{2}}$ for any $s>0$ and $c_j\geq0$), or
$f_i$ is truncated exponential and $\mu_i(1 -P_i+c_i)\exp\left(-P_i -1+P_{j}+s\right) \leq \mu_j\exp\left(-P_i-1+P_{j}-s\right)$. 

By the Implicit Function Theorem, $\frac{\partial P_i^*}{\partial s} =-\frac{\partial FOC_i^*}{\partial s}\left/\frac{\partial FOC_i^*}{\partial P_i}\right.$. Sufficient for the main result $\frac{dP_i^*}{ds}<0$ is $\frac{\partial FOC_i^*}{\partial s} <0$, because the SOC implies $\frac{\partial FOC_i^*}{\partial P_i} <0$ and prices are strategic complements by Lemma~\ref{lem:pure}. 
\end{proof}

The intuition for the main result is that the fraction of switchers among a firm's customers falls in the search cost. 
The switchers can be held up, because they are willing to pay the price plus the search cost, thus will all still buy if the firm's chosen price exceeds the expected price by less than the search cost. The hold-up motive increases a firm's optimal price. Greater search cost decreases the hold-up motive, thus the price. When the search cost becomes so large that no consumers switch, each firm's price falls to its monopoly level. 
The monopoly price is with respect to the remaining demand at the large search cost when low-valuation customers have exited. This demand is smaller than at lower search cost and contains relatively more high-valuation customers. Therefore the monopoly price at the remaining demand is greater than for a joint owner of the firms at a smaller search cost. 

\begin{figure}
\caption{Demands after an increase in the search cost from $0.1$ to $0.2$ at $P_{\x}=0.6$, $P_{\y}=0.45$.}
\label{fig:srises}
\includegraphics[width=\linewidth]{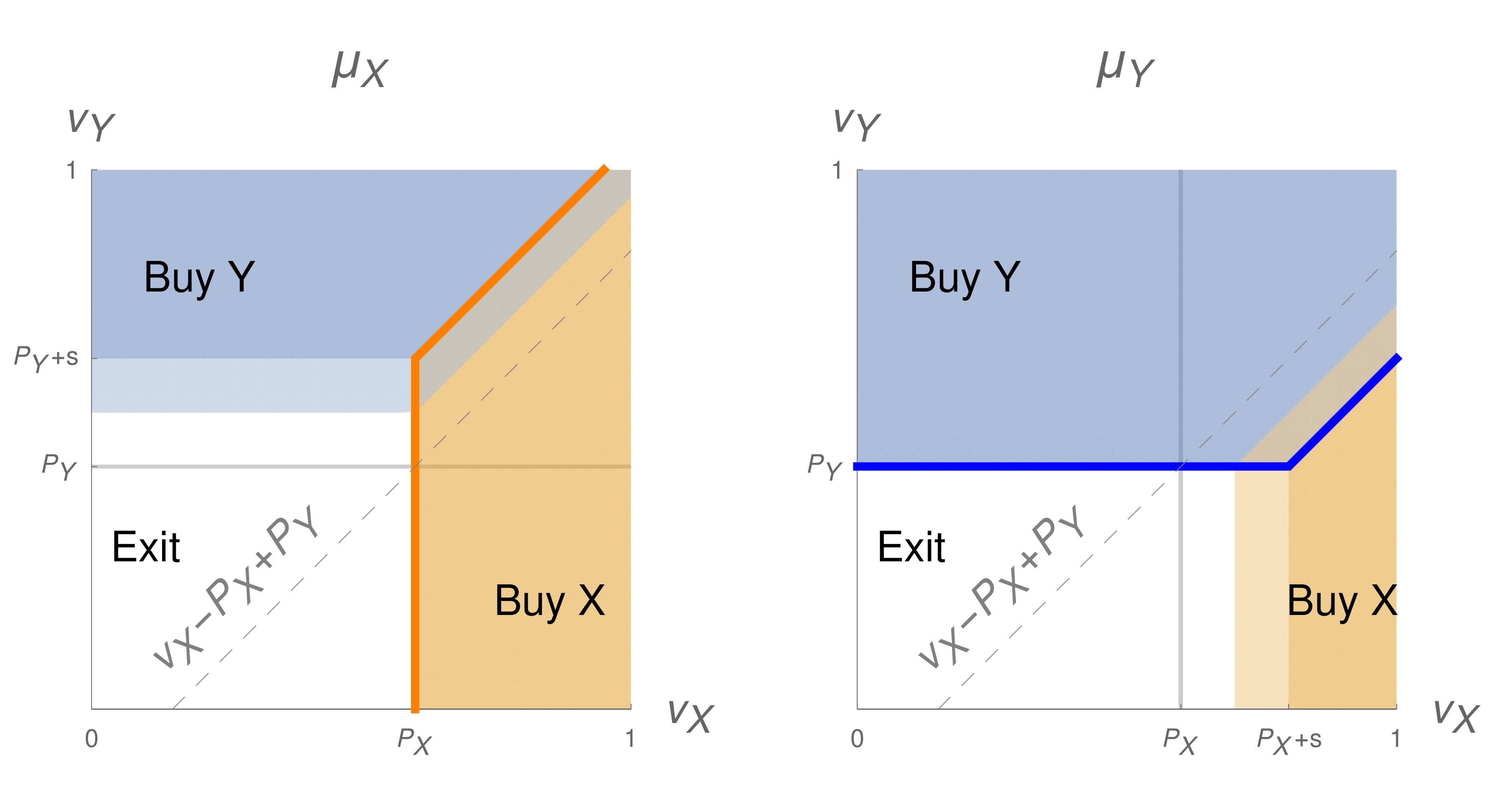} 
\end{figure}
Figure~\ref{fig:srises} shows the effect of a greater search cost on the demands, fixing the prices. In the left panel (consumers initially at firm $\x$), the light blue rectangle denotes the consumers who stop buying from $\y$ and exit when $s$ increases. The bluish diagonal band below the orange line are consumers who start buying from $\x$ instead of $\y$. In the right panel, the orange rectangle consists of consumers who stop buying from $\x$ and exit, while the orange diagonal band below the blue line depicts those who stay with $\y$ instead of switching to $\x$. Each firm loses some switchers who could be held up and gains some demand from consumers initially at itself who respond to price increases. 

This concludes the discussion of how the search cost affects prices. 
The following subsection examines the change of profits, welfare and consumer surplus in the search cost, as well as how prices respond to production costs and initial demands.

\subsection{Other comparative statics}

Exit increases in the search cost, so total trading surplus in the market falls. Costlier search also makes the final allocation of buyers to firms less efficient (ideally, consumers above the diagonal would buy from firm $\y$ and below the diagonal from $\x$). Therefore welfare and profits decrease in the search cost. The effect on consumer surplus could have either sign, because both prices and allocative efficiency decrease.

The comparative statics in production costs $c_i$ are intuitive: higher costs raise prices. The FOC~(\ref{foc}) increases in $c_i$, so by \cite{milgrom+roberts1990} Theorem~6, $P_i$ increases in $c_i$. Strategic complementarities then imply that all prices rise in any firm's cost. 

The effect of greater asymmetry of initial demands (higher $|\mu_{\x}-\mu_{\y}|$) on prices is ambiguous, because the best response curves of the firms move in opposite directions when $\mu_{\x}-\mu_{\y}$ increases. 
Each firm's best response curve may in general rise or fall in the amount of asymmetry. 
Strategic complementarities mean that both best responses are increasing in the rival firm's price. The intersection of two increasing functions may move in any direction when one function pointwise increases and the other decreases. 
When the firms are symmetric, the best response of firm $i$ decreases in $\mu_{i}-\mu_{j}$, which rules out a simultaneous increase in $P_i$ and decrease in $P_j$. The prices may jointly increase or decrease, or a fall in $P_i$ may co-occur with an increase in $P_j$. 

Suppose the initial demands are maximally asymmetric---all consumers are initially at one firm (the incumbent). Then the unique equilibrium prices are such that the incumbent remains a monopolist. The consumers expect a high enough price from the other firm (the entrant) that learning is not optimal. Any price cut by the entrant is not observed by consumers, so they cannot start learning in response to it. Suppose consumers expected the entrant to set a low enough price to make learning worthwhile at some valuations. Then the entrant would hold up all arriving switchers by choosing a price greater than they expected, for any expected price and any positive search cost. This contradicts consumers learning the entrant's price. 

Modifications of the baseline model are considered next. The results remain robust to a distribution of search costs, unattached consumers or many firms, and are continuous in the correlation of valuations.

\section{Extensions and generalisations}
\label{sec:extensions}

A distribution of search costs bounded away from zero and independent of the valuations yields the same results as the mean search cost, because consumers are risk neutral. With unchanged consumer decisions, the profit-maximising prices remain the same. The comparative statics use the mean of the search cost distribution and are otherwise unchanged. 

If some consumers have zero search cost and some positive, then firms mix over an interval of prices. Firms want to undercut each other to attract zero-cost buyers, but if prices are cut low enough, then prefer to charge a high price and sell just to their initial positive-cost customers. The hold-up motive is still present, due to the positive-cost buyers, and decreases in the search cost, so the direction of the comparative statics remains the same. 

With zero search cost for all buyers, prices are discretely lower than with a small positive search cost, because the mass of inframarginal consumers is continuous in $s$ even at $s=0$, but the mass of marginal consumers approximately doubles at $s=0$ compared to a small positive $s$.  Hold-up is impossible if consumers can costlessly switch firms. The discontinuity in prices at zero search cost is similar to the Diamond paradox. 

If some customers are initially at neither firm and have to pay the search cost no matter which one they go to, then the hold-up motive is strengthened for both firms. The price decrease in the search cost becomes larger.

Correlated valuations of consumers may change the results, depending on the joint distribution of the valuations. The only modification in the proofs is replacing $F_i$ in all formulas by $F_{i|j}(\cdot|v_j)$. 
These modified sufficient conditions may be harder or easier to satisfy than the original assumptions, depending on the joint distribution of the valuations. 
If $v_{\x}$ and $v_{\y}$ are perfectly positively correlated, then the model with $s=0$ is Bertrand competition, and with $s>0$, the original \cite{diamond1971} paradox, where prices stay constant in $s$. 

Perfect negative correlation of $v_{\x}$ and $v_{\y}$ reduces the environment with $s=0$ to the Hotelling model. 
With $s>0$ and $v_{\y} =1-v_{\x}$, firm $i$'s FOC is 
$\frac{\mu_i}{\sqrt{2}}(P_{j^*}+s-P_i) +\frac{\mu_j}{\sqrt{2}}(P_j^*-P_i^*-s)-(P_i-c_i)\left[\frac{\mu_i}{\sqrt{2}} +\frac{\mu_j}{\sqrt{2}}\text{\textbf{1}}\set{P_i>P_{i}^*+s}\right] =0$. 
Equilibrium prices are 
$P_{i} =\frac{\mu_j(\mu_i-\mu_j)s +(1+\mu_j)\mu_ic_i +\mu_jc_j}{(1+\mu_i )(1+\mu_j)-1}$.
The marginal consumers are independent of $s$ as long as $P_i<P_{i}^*+s$. The inframarginal buyers change by $\frac{\mu_i-\mu_j}{\sqrt{2}}$ when $s$ increases by a unit. Firms with symmetric initial demands thus charge the same prices for any $s>0$. A firm with more initial customers raises its price in $s$, but the smaller firm decreases (and faster). With equal costs $c_i=c_j$, the demand-weighted average price $\mu_i P_i +\mu_j P_j$ stays constant in $s$. 

Consumers uniformly distributed on a + symbol in the demand squares result in both firms charging a constant price for all $s\geq0$. Thus each firm's competitive and monopoly price are equal. The price of firm $\y$ is equal to the distance from the x-axis to the centre of the +, and symmetrically for firm $\x$. An example of such a distribution is two crossing streets with consumers living along their length. 

Many firms are conceptually similar to duopoly---in each firm's FOC, replace the rival firm with the combination of all rivals. 
The incentives of firm $i$ are the same as when facing a single competitor which has initial demand $\sum_{j\neq i}\mu_j$ and offers consumers the net value $\max_{j\neq i}\set{v_{j}-P_{j}^*}$ distributed according to $\prod_{j\neq i}F_j(\cdot +P_j^*)$. Equilibrium prices in an oligopoly are of course lower than in a duopoly in which all rivals are controlled by a single owner. However, the comparative statics retain their direction, because the FOC of each firm still decreases in $s$ for a range of parameters. 

%

If consumers initially at firm $i$ do not know their valuation $v_j$ for the rival firm, but can learn $v_j$ and $P_j$ together, then valuation distributions close to uniform result in intuitive comparative statics---prices increase in the search cost.\footnote{If $f_i=1$, $c_i=0$ and $\mu_i=\frac{1}{2}$, then the monopoly prices at large $s$ are $P_i=\frac{1}{2}$ and the competitive prices at $s=0$ are $P_i \approx 0.414$.
} 
However, for symmetric firms, a sufficiently large decrease in the valuation pdf at a point $v^*$ above the monopoly price implies that prices decrease in the search cost over some range of $s$. 
A numerical example has $\mu_{i}=\frac{1}{2}$, $c_{i}=0$ and $f_{i}(v_i)=\begin{cases}
\frac{3}{2} & \text{if } v_i\in[0,\frac{1}{2}], \\
\frac{1}{2} & \text{if } v_i\in(\frac{1}{2},1], \\
\end{cases}$ for both firms. Thus $v^*=\frac{1}{2}$. The equilibrium prices at $s=0$ are approximately $0.31$, and the monopoly price as $s$ becomes large is $0.25$. 
As $s$ increases from $0.13$ to $0.19$, prices decrease linearly from $0.491$ to $0.384$. 

At large $s$, the environments with known and unknown valuation for the other firm are identical because no consumers switch. At $s=0$, these models are also identical, because it is weakly dominant for all consumers to search. If consumers know their valuations, then prices jump up when the search cost becomes positive, but if the valuation for the other firm is unknown, then prices are continuous at costless search. Thus for low positive search costs, consumers obtain greater utility when they do not know their valuation. Firms correspondingly make lower profits.

The following section concludes with a discussion of the predictions and policy implications from the main model.

\section{Discussion}
\label{sec:discussion}

When prices and profits decrease in the search cost, industry associations naturally want to provide information that helps customers compare members. An online directory achieves this, which justifies the cost of creating and maintaining the member database. Notably, such searchable directories do not provide price comparisons, even though these would be easy to add. A simple explanation why not is that reducing the search cost to zero by making prices transparent would discretely decrease prices and profits relative to a small positive cost. 

At low positive search costs, prices and profits are discretely higher when consumers know their valuation for each firm before the learning decision than when they learn the valuation together with the price. This is an additional motive for industry groups to inform consumers in detail about the goods and services each member provides. 

For a large enough search cost, each firm is a monopolist over its initial customers, which would be a reason for high prices, especially when the exit of low-valuation buyers increases the average willingness to pay among the remaining ones. However, the exit of many consumers (who are inefficiently allocated to a firm which they value little) reduces total surplus enough to outweigh the larger share of surplus that a monopolist can obtain using its market power. Therefore firms prefer a more efficient allocation even if it means more competition. 

As Adam Smith already noted, industry associations tend to collude to increase the profits of their members at the expense of consumers. A regulator maximising consumer surplus prefers either zero search cost, or if this is unattainable, then maximal cost. Prohibiting information release by an association is difficult, so the regulator should instead provide price comparisons directly. Examples are a government-run health insurance exchange and a government website listing pension funds ordered by their total fee loading. Of course, the industry can counter by obfuscating prices with hidden add-on costs and private discounts. 

A regulator maximising total surplus unambiguously prefers a lower search cost. At small positive search costs, both kinds of regulator prefer that consumers do not know their valuation for the rival firm. However, providing price information to consumers dominates removing their valuation information even if the latter was possible.


\bibliographystyle{ecta}
\bibliography{teooriaPaberid} 
\end{document}